\documentclass{uai2023} 

\usepackage[american]{babel}
\usepackage{graphicx}
\usepackage{subfigure}
\usepackage{natbib} 
\usepackage{mathtools} 
\usepackage{booktabs} 
\usepackage{tikz} 
\usepackage{amsfonts}
\usepackage{algorithm}
\usepackage{algorithmic}
\usepackage{changes}
\usepackage{mathtools}
\usepackage{subfigure}
\usepackage{multirow}
\usepackage{makecell}
\usepackage{amsfonts}
\usepackage{booktabs}
\usepackage{amsmath,amssymb,amsthm}

\newtheorem{definition}{Definition}

\newtheorem{theorem}{Theorem}

\newtheorem{lemma}{Lemma}
\usepackage{todonotes}
\title{
Enhanced Equilibria-Solving via Private Information Pre-Branch Structure in Adversarial Team Games
}

\author[1]{Chen Qiu}
\author[2]{Haobo Fu}
\author[3,4]{Kai Li}
\author[1]{Weixin Huang}
\author[1]{Jiajia Zhang}
\author[1]{Xuan Wang}
\affil[1]{%
    School of Computer Science and Technology\\
    Harbin Institute of Technology Shenzhen\\
    Shenzhen\\
    China
}
\affil[2]{%
    Tencent AI Lab\\
    Shenzhen\\
    China
}
\affil[3]{%
    Institute of Automation\\
    Chinese Academy of Sciences\\
    Beijing\\
    China
  }
\affil[4]{
    School of Artificial Intelligence\\
    University of Chinese Academy of Sciences\\
    Beijing\\
    China
}

\begin{document}

\maketitle
\thispagestyle{empty}
\begin{abstract}
In \emph{ex ante} coordinated adversarial team games (ATGs), a team competes against an adversary, and the team members are only allowed to coordinate their strategies before the game starts. The team-maxmin equilibrium with correlation (TMECor) is a suitable solution concept for ATGs. One class of TMECor-solving methods transforms the problem into solving NE in two-player zero-sum games, leveraging well-established tools for the latter. However, existing methods are fundamentally action-based, resulting in poor generalizability and low solving efficiency due to the exponential growth in the size of the transformed game. To address the above issues, we propose an efficient game transformation method based on private information, where all team members are represented by a single coordinator. We designed a structure called \emph{private information pre-branch}, which makes decisions considering all possible private information from teammates. We prove that the size of the game transformed by our method is exponentially reduced compared to the current state-of-the-art. Moreover, we demonstrate equilibria equivalence. Experimentally, our method achieves a significant speedup of 182.89$\times$ to 694.44$\times$ in scenarios where the current state-of-the-art method can work, such as small-scale \emph{Kuhn poker} and \emph{Leduc poker}. Furthermore, our method is applicable to larger games and those with dynamically changing private information, such as \emph{Goofspiel}.
\end{abstract}
\section{Introduction}\label{intro}
Games have long served as critical testbeds for exploring how effectively machines can make sophisticated decisions since the early days of computing \cite{DBLP:journals/ai/BardFCBLSPDMHDM20, campbell2002deep, silver2017mastering}. Finding equilibrium in games has become a significant criterion for evaluating the level of artificial intelligence. In the real world, there have been systems that have achieved superhuman performance, such as \emph{AlphaGo} \cite{silver2016mastering}, \emph{Libratus} \cite{brown2018superhuman}, and \emph{DeepStack} \cite{moravvcik2017deepstack}. While many advances \cite{brown2019deep, Zhou2020, https://doi.org/10.1002/int.21857} have been made in 2-player zero-sum (2p0s) games based on Nash equilibrium (NE) \cite{nash1951non} in imperfect information environments, recent research has focused on more complex adversarial team games (ATGs). In ATGs, multiple players with the same utility function form a team to compete against a common adversary \cite{von1997team}. This results in a game where both cooperation and competition coexist.

In this paper, we focus on the ATGs with \emph{ex ante} coordination. More specifically, team members are allowed to coordinate and agree on a common strategy before the game starts. The solution concept for this setting is the team-maxmin equilibrium with correlation (TMECor), which can be thought of as an NE between the team and the adversary in an ATG \cite{Zhang_2021}. TMECor has the properties of existence and uniqueness in the ATGs with \emph{ex ante} coordination, which avoids the equilibrium selection problem. However, finding a TMECor is proven to be FNP-hard \cite{hansen2008approximability}. 

Currently, methods for computing TMECor can be roughly divided into three categories. The first involves using linear programming (LP). Hybrid column generation was the first algorithm to compute TMECor in ATGs \cite{celli2018computational}. Its core involves team members adopting joint normal-form strategies, while the adversary uses sequence-form strategies. The main weaknesses of such methods are the necessity to solve an integer LP  and the exponential growth in the number of joint actions by the team as the size of the game increases, making them impractical for large-scale games. The second category involves multi-agent deep reinforcement learning algorithms, which can learn coordinated strategies from experience without prior knowledge, such as the method proposed by \citet{DBLP:conf/atal/CacciamaniCC021}. However, it is only applicable to games with symmetric observations for the team, as it requires perfect recall refinement. When team members in the game have private information and public actions, perfect recall refinement is not applicable. The last category of methods combines game tree transformation with techniques commonly used in 2p0s games, establishing a connection between ATGs and 2p0s games. Additionally, the team's strategies are highly explainable, since they are behavioral over the game tree with a direct interpretation. In this paper, we focus on the game tree transformation-based approach.

The state-of-the-art method in game tree transformation-based method is TPICA, proposed by \citet{carminati2022marriage}, which involves the concept of \emph{extensive-form game with visibility}. By introducing the coordinator into an ATG, they convert the task of finding a TMECor into finding an NE in a 2p0s game. However, this method suffers from low solving efficiency and limited types of solvable games. The primary reason is that TPICA relies on action-based transformation. In the transformed game, the coordinator extracts an action from each distinguishable state in the original game to form different recommendations for the team players. Then, a specific action is designated from each recommendation as an available action for the coordinator. To objectively analyze the game size complexity, we consider a setting where the opponent plays first, followed by team members in sequence. We assume that every player has the same number of available actions in every state. Since modifications are made only to team player nodes, we refer to the phase from any specified team player node, which acts first, to all possible next opponent (or terminal) nodes as an episode. For any episode, the size of the transformed game tree is $\mathcal{O}\big((\lvert A\rvert^{\lvert \Omega\rvert})^{\lvert \mathcal{T}\rvert}\big)$, where $\lvert A\rvert$, $\lvert \Omega\rvert$, $\lvert \mathcal{T}\rvert$ denote the number of available actions, the amount of private information, and the number of team players, respectively. Therefore, the size of the transformed game tree grows significantly with the increase in the number of available actions, team players, and private information. In particular, the size growth triggered by a single coordinator node is exponential. Additionally, TPICA cannot be applied to games where players' private information changes.

To address the above issues, we propose a multi-player transformation algorithm (MPTA) based on private information. In our method, to mitigate the exponential growth caused by a single coordinator, we designed a new structure called \emph{private information pre-branch} (PIPB), which consists of coordinator nodes and dummy player nodes. Specifically, PIPB allows dummy players to provide the coordinator with all possible private information from teammates. Since the amount of potentially private information in an ATG is fixed, this structure significantly reduces the size of the transformed game tree compared to the previous state-of-the-art method. This leads to a substantial improvement in the efficiency of equilibrium computation. Furthermore, we demonstrate the equilibrium equivalence before and after the transformation. The private information-based transformation makes our method suitable for games with dynamically changing private information (e.g., \emph{Goofspiel}), expanding the types of solvable games. We show the superior performance of our method through multiple experiments in different game scenarios. The experimental results show that our method computes strategies closer to TMECor compared to the baseline algorithm in the same runtime and significantly reduces runtime within the same number of iterations.

Our contributions can be summarized as follows:
\begin{itemize}
        \item We proposed MPTA based on private information, which significantly improves equilibria-solving efficiency. For any episode, compared to the previous state-of-the-art, the size growth is reduced from $\mathcal{O}\big((\lvert A\rvert^{\lvert \Omega\rvert})^{\lvert \mathcal{T}\rvert}\big)$ to $\mathcal{O}\big((\frac{(\lvert \Omega\rvert-1)!}{(\lvert \Omega\rvert-\lvert \mathcal{T}\rvert)!}\lvert A\rvert)^{\lvert \mathcal{T}\rvert}\big)$, where $\frac{(\lvert \Omega\rvert-1)!}{(\lvert \Omega\rvert-\lvert \mathcal{T}\rvert)!}$ represents the number of ways to arrange $\lvert \mathcal{T}\rvert-1$ elements from a set of $\lvert \Omega\rvert-1$ private information. Additionally, we demonstrated the equilibria equivalence between TMECor in the original game and NE in the transformed game.
        
        \item  The PIPB structure enhances the generalization capability of our method and expands the types of solvable games. It allows our method to be applied to ATGs where players' private information dynamically changes.
        
        \item We conducted 14 experiments on three standard testbeds. The results show a significant improvement in solving efficiency using our method, achieving speedups ranging from 182.89$\times$ to 694.44$\times$ compared to the baseline. We also compared the sizes of the transformed game trees, showing that our method results in much smaller game trees than the baseline. Furthermore, we experimented with larger-scale games and other types of games that were not supported by the baseline.
    \end{itemize}
All proofs in this paper can be found in Appendix A.

\section{Related Work}
Significant research has been focused on finding suitable solutions for ATGs since the concept of team-maxmin equilibrium was introduced by \citet{von1997team}. According to the communication capabilities of the team members, \citet{celli2018computational} defined three different scenarios and corresponding equilibriums for the first time in the extensive-form ATGs.

\citet{basilico2017computing} proposed a modified version of the quasi-polynomial time algorithm and a novel anytime approximation algorithm named \emph{IteratedLP}, whose working principle is to maintain the current solution, providing a policy that can be returned at any time for each team member. \citet{farina2018ex} adopted a novel realization-form representation that maps the problem of finding an optimal ex-ante-coordinated policy for the team to the problem of finding NE. \citet{zhang2020computing} investigated the computational inefficiency resulting from the correlation between team members' strategies and proposed an associated recursive asynchronous multiparametric disaggregation technique to accelerate the computation of TMECor. They accomplished this by reducing the solution space of a mixed integer linear program using an association constraint. Successively, \citet{zhang2020converging, Zhang_2021, farina2021connecting} proposed more efficient variants of the LP. Although \citet{zhang2022team} used a tree decomposition for constraints and described the team's strategy space by a polytope, finding TMECor still requires solving an LP. Some researchers have attempted to use multi-agent deep reinforcement learning to handle ATGs. For instance, \citet{DBLP:conf/atal/CacciamaniCC021} added a game-theoretic centralized training regimen and served as a buffer of past experiences. However, this method can only be applied to games where team members have symmetric observations of each other. It cannot be extended to general games with private information and public actions, such as poker.

The idea of using a coordinator to coordinate team members can be traced back to the seminal work of decentralized stochastic control \cite{nayyar2013decentralized}. The TPICA proposed by \citet{carminati2022marriage} is closely related to our method. Although TPICA has strong theoretical guarantees, their action-based game transformation method results in exponential growth in game size. Compared to TPICA, our method not only offers the same theoretical guarantees but also significantly reduces game size, greatly improving the efficiency of computing TMECor. Moreover, our method expands the types of solvable games, primarily due to the designed PIPB structure.

\section{Preliminaries}\label{preli}
\subsection{Extensive-Form Games and Nash Equilibrium}
An extensive-form game (EFG) is the tree-form model of imperfect-information games with sequential interactions \cite{kuhn1950extensive, brown2017safe, https://doi.org/10.1002/int.22450}. The set $A=\cup_{i \in N} A_i$ denotes all the possible actions, where $A_i$ represents a set of available actions of player $i$. $\lvert A\rvert$ is the number of each player's available actions. $H$ is the set of nodes, and $h\in H$ represents the sequence of all actions from the root to node $h$. $ha\sqsubseteq h^{\prime}$ denotes $h$ reaches $h^{\prime}$ by playing an action $a$. The set $Z \subseteq H$ contains all the terminal nodes. For each decision node $h \in H$, the result returned by the function $\mathcal{A}(h)$ is all available actions at node $h$. $\omega_{i}\in \Omega$ denotes player $i$'s private information (e.g., a card in a poker game), and $\lvert \Omega\rvert$ represents the amount of private information in a game. The player who takes an action at node $h$ is returned by function $P(h)$. The utility function $u_{i}(z)$ is the player $i$'s payoff mapped from a terminal node $z \in Z$ to the real $\mathbb{R}$. An information set (infoset) $I_i$ represents imperfect information for player $i$, which means all nodes $h, h^{\prime}$ are indistinguishable to $i$ in $I_i$. The set of infosets for player $i$ is denoted by $\mathcal{I}_{i}$, and the set of all infosets is represented as $\mathcal{I}=\cup_{i\in N}\mathcal{I}_{i}$.

There are two fundamental paradigms for strategy representation \cite{carminati2022marriage,https://doi.org/10.1002/int.21950}. A behavioral strategy $\sigma_{i}$ of player $i\in N$ is a function that assigns a distribution over all the available actions $\mathcal{A}\left(I_i\right)$ to each $I_i$. Another strategy representation is based on the normal-form plan (a.k.a. pure strategy) $\pi_{i} = \times_{I\in \mathcal{I}_{i}} \mathcal{A}(I)$ which is a tuple specifying on action for each infoset of player $i$. A normal-form strategy is the probability distribution of normal-form plans for a player.  A reduced normal-form strategy (a.k.a. mixed strategy) $\mu_{i}\in \Delta\left(\Pi_i\right)$ is obtained from a normal-form strategy by consolidating plans that are differentiated via actions taken in unreachable nodes. Henceforth, we focus on reduced normal-form strategies in this paper. For any player $i\in N$, $\mu_{i}[z]$ (or $\sigma_{i}[z]$) denotes the probability of reaching terminal nodes $z\in Z$ when $i$ follows strategy $\mu_{i}$ (or $\sigma_{i}$). We represent behavioral strategy profiles with $\boldsymbol{\sigma}$ and normal-form strategy profiles with $\boldsymbol{\mu}$. We define $\boldsymbol{\sigma}_{-i}$ as strategies of players except for $i$. The expected payoff for player $i$ when he plays strategy $\sigma_{i}$ and all the other players follow strategy $\boldsymbol{\sigma}_{-i}$ is denoted by $u_{i}(\sigma_{i},\boldsymbol{\sigma}_{-i})$. Player $i$'s \emph{best response} to strategy $\sigma_{-i}$, denoted as $BR_{i}(\boldsymbol{\sigma}_{-i})$, is a strategy that maximizes player $i$'s payoff against strategy $\boldsymbol{\sigma}_{-i}$. Formally, $u_{i}(BR_{i}(\boldsymbol{\sigma}_{-i}), \boldsymbol{\sigma}_{-i})=max_{\sigma_{i}^{\prime}}u_{i}(\sigma_{i}^{\prime},\boldsymbol{\sigma}_{-i})$. NE is a significant solution concept in 2p0s games in which no player can unilaterally change his strategy to obtain more payoff. An NE $\boldsymbol{\sigma}$ is a strategy profile where all players play the \emph{best response}. Formally, $\boldsymbol{\sigma}$ is an NE if and only if $\forall i\in N, \sigma_{p}\in BR_{i}(\boldsymbol{\sigma}_{-i})$ The \emph{exploitability} $e(\sigma_i)$ of strategy $\sigma_i$ serves as our measurement metric, which measures how much worse $\sigma_i$ does versus $BR(\sigma_i)$ compared to how an equilibrium strategy $\sigma_i^{*}$ does against $BR(\sigma_i^{*})$.

\subsection{Adversarial Team Games and Team-Maxmin Equilibrium with Correlation}

An ATG is an EFG with a set of players $N$, where a team of players competes against an opponent. That is, $N=\mathcal{T}\cup \left\{o\right\}\cup \left\{c\right\}$, where $\mathcal{T}$ represent a team, and $o$ is an opponent. The chance player $c$ simulates exogenous randomness in the game, such as dealing a card from a deck. $\lvert \mathcal{T}\rvert$ denotes the number of team members. The team players share payoffs in ATGs. Formally, $\forall i,j\in \mathcal{T}, u_{i}(z)=u_{j}(z)$. Following the convention of the relevant literature \cite{celli2018computational, Zhang_2021, carminati2022marriage}, we assume \emph{perfect recall}, which means each player remembers information acquired in earlier stages of each infoset.

\begin{algorithm}[tb]
\caption{Multi-Player Transformation Algorithm}
\label{algo:MPTA}
\begin{algorithmic}[1] 
\STATE \textbf{Function} \emph{MPTA}($G$)
\STATE initialize $G^{\prime}$ 
\STATE $N \gets \mathcal{T}\cup \left\{o\right\}\cup \left\{c\right\}$
\STATE $N^{\prime} \gets \{t\} \cup \{o\} \cup \{c\}$
\STATE initialize $h$ with the chance player node
\STATE $h^{\prime} \gets$ \emph{ProcOfTrans}$(G, G^{\prime}, h)$ 
\RETURN{$G^{\prime}$}

\STATE \textbf{Function} \emph{ProcOfTrans}($G, G^{\prime}, h$)
\STATE $\Omega \gets$ Private$(G)$  
\IF{$P(h) = c$}
\STATE $h^{\prime} \gets h$
\STATE $\mathcal{A}^{\prime}(h^{\prime}) \gets \mathcal{A}(h)$
\FOR{$a^{\prime}\in \mathcal{A}^{\prime}(h^{\prime})$}
\STATE \emph{ProcOfTrans}($G, G^{\prime}, ha^{\prime}$)
\ENDFOR
\ELSIF{$P(h)=o$}
\STATE $h^{\prime} \gets h$
\STATE $\mathcal{A}^{\prime}(h^{\prime}) \gets \mathcal{A}(h)$
\ELSIF{$P(h)\in {\mathcal{T}}$}
\STATE add a dummy player node $h_d$ as $h$'s parent node
\STATE $\mathcal{A}^{\prime}(h_d) \gets \Omega\setminus \{\omega_{P(h)}\}$
\FOR{$a^{\prime} \in \mathcal{A}^{\prime}(h_d)$}
\STATE $h^{\prime}\gets h_{d}a^{\prime}$
\FOR{$a\in \mathcal{A}(h)$}
\STATE ProcOfTrans($G,G^{\prime},h^{\prime}a$)
\ENDFOR
\ENDFOR
\ELSE
\STATE $z^{\prime} \gets h$
\STATE $u_{t}(z^{\prime}) \gets \sum_{i\in\mathcal{T}}u_{i}(h)$
\STATE $u_{o}^{\prime}(z^{\prime}) \gets -u_{t}(z^{\prime})$
\ENDIF
\RETURN $h^{\prime}$
\end{algorithmic}
\end{algorithm}

In this work, we focus on the \emph{ex ante} coordinated setting, where TMECor is a significant solution concept. Specifically, a TMECor is an NE that maximizes the team's payoff when team players are allowed to correlate their strategies and agree on tactics before the game begins. A TMECor can be found via a bi-level optimization program formulated over the normal-form strategy profile of team members:
\begin{equation} \label{equ:TMECor}
\begin{split}
\max_{\mu_{\mathcal{T}}}  \min _{\mu_{o}} & \sum_{z \in Z} \mu_{\mathcal{T}}[z]\mu_{o}[z]u_{\mathcal{T}}(z) \\  \text{s.t.} \hspace{.5cm} & \mu_{\mathcal{T}} \in \Delta(\times_{i\in\mathcal{T}}\Pi_i) \\ &  \mu_{o} \in \Delta\left(\Pi_{o}\right)
\end{split}
\end{equation}

\subsection{Team-Public-Information Representation for Extensive-Form Games}
Since this subsection involves some additional concepts, we provide a detailed example in Appendix B for a clearer explanation.
An action $a$ is classified as \emph{observable} or \emph{unobservable} depending on whether it can be seen by player $i$ when played by another player. If the actions $a$ observable by player $i$ at any pair of nodes are the same, these nodes belong to the same infoset. When all infosets are induced, as discussed above, the game is an extensive-form game with visibility (vEFG), where every player has perfect recall. Extending action visibility to a set of player $\mathcal{P}$ (e.g., a team), an action $a$ is called \emph{public} if it is observable by all players in $\mathcal{P}$; \emph{private} if it can be observed by only some player(s) in $\mathcal{P}$; and \emph{hidden} if it is not observable by all players in $\mathcal{P}$ (in this case, $a$ is played by a player not belonging to $\mathcal{P}$). A public infoset for a set of players $\mathcal{P}$ is defined as a public state $S_{\mathcal{P}}\subset H$, where any two nodes of potentially different players in $\mathcal{P}$ belong to the same public state if the actions that are \emph{public} for $\mathcal{P}$ are the same at these nodes. Clearly, if one node of an infoset $I$ belongs to $S_{\mathcal{P}}$, then all the other nodes of $I$ also belong to $S_{\mathcal{P}}$. Let $\mathcal{S}$ denote the set of all public states. $\mathcal{S}_{\mathcal{P}}(h)$ represents the set of all infosets of players in $\mathcal{P}$ that are in the same public state at node $h$.

\begin{figure*}[t]
\centering
\includegraphics[width=1\textwidth]{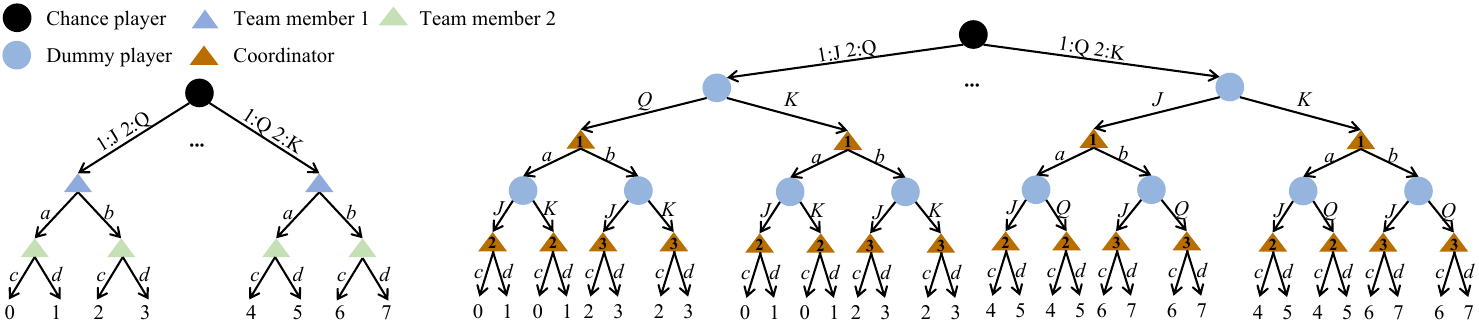}
\caption{Example of game transformation. ``\textbf{\dots}'' indicates omitted branches. The nodes of a player with the same number are in the same infoset. \textbf{Left:} Original ATG omitting the opponent. \textbf{Right:} Result of transforming the game on the left using MPTA.}
\label{fig:TransformedMPTA}
\end{figure*}

Similarly, this paper focuses on public-turn-taking games, where every player knows, at every infoset he plays, the sequence of actions taken by players from the root to that infoset. This indicates that the public states have a specific structure and consist of nodes with histories of the same length for a single player. \citet{carminati2022marriage} proved that any vEFG with perfect recall and timeability has a strategy-equivalent public-turn-taking vEFG. Completely inflated games refer to situations where every team player knows the exact action taken by another team player at any infoset. This can be achieved for a generic vEFG by modifying the visibility of team players' actions, allowing explicit representation of strategy sharing among teammates before the game starts. In the following, we focus on completely inflated vEFGs for the team.

In the game tree transformed by TPICA, every coordinator who represents the team $\mathcal{T}$ plays a \emph{prescription} among all the combinations of possible actions for each information state $I$ belonging to the public team state. In other words, for a public team state $S$, the coordinator issues different recommendations to players for every possible information set associated with $S$. Then, dummy players are used to extract a specific action for each infoset from prescriptions and pass it to the next player. Whenever it is the coordinator's turn to play, all the possible action combinations are listed. Therefore, during the transformation, every added dummy player node corresponds to another node. For any episode, the size of the game tree transformed by TPICA is $\mathcal{O}\big((\lvert A\rvert^{\lvert \Omega\rvert})^{\lvert \mathcal{T}\rvert}\big)$.

\section{Method}\label{method}
In this section, we provide a comprehensive introduction to our method. We begin by designing a new structure called \emph{private information pre-branch} (PIPB), which provides the coordinator with all possible private information from teammates. This effectively reduces the size of the transformed game since the amount of potentially private information is fixed in a game. Then, we leverage this structure to propose a multi-player transformation algorithm (MPTA) based on private information. It effectively transforms an ATG into a strategy-equivalent 2p0s game and expands the types of solvable games. Additionally, we provide a proof of equilibrium equivalence. Finally, we demonstrate that the size of the game tree transformed by our method is smaller than that of previous state-of-the-art algorithms. 

\subsection{The Structure of \emph{Private Information Pre-Branch}}

In the game transformation process, we introduce two types of players: the coordinator, representing the team, and the dummy player, who conveys teammates' possible private information to the coordinator. We consider the general case where players' private information is interdependent. For example, when a team member is dealt $J$ from a deck of three cards $J,Q,K$, he can safely infer that his teammates' private information cannot include $J$. We define the PIPB structure as follows:
\begin{definition}\label{def:PIPB}
    Given a completely inflated vEFG $\mathcal{G}$ that satisfies the public-turn-taking property, let $h_{d}$ denote a dummy player node and $H_{t} \subseteq H$ denote a set of nodes of the same team member. $h_{d}$ and $H_{t}$ form a PIPB iff $\forall h\in H_{t}, \forall \omega\in \Omega\setminus \{\omega_{i}\}: h_{d}\omega\sqsubseteq h$, where $i= P(h)$.
\end{definition}

Intuitively, all team member nodes in a PIPB are connected to a single dummy player node, which serves as their common parent node. The dummy player's available actions represent all possible private information of the teammates. These actions are \emph{unobservable} to all players except those in the next layer. After the transformation, all team players are replaced by the coordinator. While the visibility of the dummy player's actions remains unchanged, the coordinator's public states will change.

\subsection{Multi-Player Transformation Algorithm}
We propose a multi-player transformation algorithm that utilizes the PIPB structure to transform an ATG into a 2p0s game. This method achieves the equivalent transformation by using potentially private information of teammates. The pseudocode is provided in Algorithm~\ref{algo:MPTA}. Our method relies on the tree-form structure, so we first construct a complete game tree for the original ATG and then traverse it in a depth-first pre-order manner. 

To illustrate more clearly, we provide an example of an ATG, ignoring the opponent, transformed via our method, as shown in Figure~\ref{fig:TransformedMPTA}. The converted result by TPICA can be found in Appendix B. In this example, the chance player deals one card to each player from a deck of three cards ($J,Q,K$) as their private information, and team members' actions are \emph{public} (the opponent's actions are also \emph{public} if he is not ignored). The actions of the chance player are \emph{private} to team members as each team player knows only their own card. This means every team member node in the original game tree is a separate infoset. As shown in the left part of Figure~\ref{fig:TransformedMPTA}, the set of all private information $\Omega$ is $\{J,Q,K\}$. When traversing to a team member node, a dummy player node is first added. This dummy player's actions represent all possible private information of the teammates, i.e., $\{Q,K\}$ or $\{J,K\}$. The private information of team members cannot be passed along since they are not publicly observable. The dummy player's actions are based on the private information of the specific team member, so they also cannot be passed along and can only be observed by team members at the next level. Then, we introduce coordinator nodes to make decisions in place of the team member nodes. Terminal nodes are copied according to the sequence of publicly observable actions in the original game. In particular, the payoff of a coordinator is the sum of all team members' payoffs.

Compared to TPICA, the PIPB structure in our method reduces the number of infosets in the transformed game. Coordinator nodes that share the same dummy player node belong to the same information set. Furthermore, the partition of the coordinator nodes' infosets is also based on actions observable to all team members (i.e., actions called \emph{public}). We divide the coordinator's public infosets according to the concept of public state, as shown in the right part of Figure~\ref{fig:TransformedMPTA}, where nodes of a player with the same number belong to the same public infoset.
\begin{theorem} \label{theorem:growth}
    Given an ATG $G$ with visibility that satisfies the public-turn-taking property, and its transformed game $G^{\prime}=\emph{MPTA}(G)$. The size of any episode in $G^{\prime}$ is $\mathcal{O}\big((\frac{(\lvert \Omega\rvert-1)!}{(\lvert \Omega\rvert-\lvert \mathcal{T}\rvert)!}\lvert A\rvert)^{\lvert \mathcal{T}\rvert}\big)$.
\end{theorem}

\subsection{Equilibrium Equivalence}
As an important theoretical guarantee for this work, we prove that the TMECor in the original game can be obtained by solving the NE in the transformed game. We call two strategies realization-equivalent if they induce the same probabilities for reaching nodes for all strategies of other players. In other words, the strategies of NE in the transformed game and the strategies of TMECor in the original game are realization-equivalent. We state Lemma~\ref{lemma:strategy} on strategy equivalence as follows.

\begin{lemma}\label{lemma:strategy}
    Given an ATG $G$ with visibility that satisfies the public-turn-taking property, and the transformed game $G^{\prime}=\emph{MPTA}(G)$. For any joint reduced pure strategy $\pi_{\mathcal{T}}$ in $G$, it can be mapped to a corresponding strategy $\pi_{t}$ in $G^{\prime}$, and vice versa.
\end{lemma}

By Lemma~\ref{lemma:strategy}, we state Theorem~\ref{theorem:payoff} on payoff equivalence as follows.

\begin{theorem} \label{theorem:payoff}
    Given a public-turn-taking ATG $G$ with visibility, and its transformed game $G^{\prime}=\emph{MPTA}(G)$, they have equivalent payoffs.
\end{theorem}

\setlength{\tabcolsep}{1mm}
\begin{table*}
  \centering
  \caption{Experimental results of the running time of TPICA and our method on several different types and sizes of game instances. Blank cells indicate that the experiment cannot be conducted.}
    \begin{tabular}{lrrrrrrrrrr}
    \toprule
    \multicolumn{1}{c}{\multirow{2}[2]{*}{\makecell[c]{Game \\ instances}}} & \multicolumn{3}{c}{Total nodes} & \multicolumn{2}{c}{Team nodes} & \multicolumn{2}{c}{Adversary nodes} & \multicolumn{2}{c}{Runtime} & \multicolumn{1}{c}{\multirow{2}[2]{*}{Improvement}} \\
         \multicolumn{1}{c|}{} & \multicolumn{1}{c}{Original} & \multicolumn{1}{c}{TPICA} & \multicolumn{1}{c|}{\textbf{MPTA}} & \multicolumn{1}{c}{TPICA} & \multicolumn{1}{c|}{\textbf{MPTA}} & \multicolumn{1}{c}{TPICA} & \multicolumn{1}{c|}{\textbf{MPTA}} & \multicolumn{1}{c}{TPICA} & \multicolumn{1}{c}{\textbf{MPTA}} &  \\
    \midrule
    12\textbf{K}3  & 151   & 5,395 & 583   & 300   & 144   & 294   & 72    & 139s  & \textbf{0.76s} & \textbf{182.89$\times$} \\
    12\textbf{K}4  & 601   & 1,337,051 & 3,097 & 3,888 & 768   & 4,632 & 384   & 1560s & \textbf{9.26s} & \textbf{168.47$\times$} \\
    12\textbf{K}6  & 3,001 & 34,191,721 & 23,161 & 261,360 & 5,760 & 368,760 & 2,880 & $>$27h  & \textbf{144s} & \textbf{694.44$\times$} \\
    13\textbf{K}6  & 23,401 &       & 271,441 &       & 75,240 &       & 22,680 &       & \textbf{562s} &  \\
    13\textbf{K}8  & 109,201 &       & 1,713,601 &       & 475,440 &       & 142,800 &       & \textbf{5,093s} &  \\
    14\textbf{K}6  & 115,921 &       & 1,796,401 &       & 528,480 &       & 105,120 &       & \textbf{3,051s} &  \\
    12\textbf{L}33 & 13,183 & 10,777,963 & 57,799 & 614,172 & 14,664 & 475,566 & 6,864 & 56,156s & \textbf{240s} & \textbf{233.98$\times$} \\
    12\textbf{L}43 & 42,589 &       & 251,749 &       & 64,008 &       & 29,736 &       & \textbf{3,006s} &  \\
    12\textbf{L}63 & 218,011 &       & 1,954,351 &       & 497,940 &       & 229,620 &       & \textbf{9,024s} &  \\
    13\textbf{L}33 & 161,491 &       & 948,151 &       & 262,500 &       & 80,220 &       & \textbf{4,014s} &  \\
    13\textbf{L}43 & 738,241 &       & 5,994,241 &       & 1,661,760 &       & 504,000 &       & \textbf{137,817s} &  \\
    14\textbf{L}33 & 1,673,311 &       & 12,226,231 &       & 3,535,320 &       & 809,880 &       & \textbf{143,475s} &  \\
    12\textbf{G}   & 2,509 &       & 92,581 &       & 29,700 &       & 1,464 &       & \textbf{241s} &  \\
    13\textbf{G}   & 15,307 &       & 3,352,669 &       & 1,107,162 &       & 13,128 &       & \textbf{50107s} &  \\
    \bottomrule
    \end{tabular}
  \label{tab_res}
\end{table*}

For brevity, we use the notation $\mapsto$ to denote the strategy mapping relationship. If we transform each pure plan and sum their probability masses, we obtain the corresponding mixed strategy. Formally, for any $\mu_{\mathcal{T}}$, the corresponding mixed strategy is $\sum_{\pi_{\mathcal{T}}:\pi_{t} \mapsto \pi_{\mathcal{T}}}\mu_{\mathcal{T}}(\pi_{\mathcal{T}})$. Therefore, Lemma~\ref{lemma:strategy} also applies to mixed strategies. Specifically, any $\mu_{\mathcal{T}}$ in $G$ can be mapped to $\mu_{t}$ in $G^{\prime}$, and vice versa.

\begin{theorem} \label{theorem:equilibrium}
    Given an ATG $G$ with visibility that satisfies the public-turn-taking property, and its transformed game $G^{\prime}=\emph{MPTA}(G)$. If $\mu_{t}^{*}$ is an NE in $G^{\prime}$, then strategy $\mu_{T}^{*}: \mu_{t}^{*}\mapsto \mu_{\mathcal{T}}^{*}$ is a TMECor in $G$.
\end{theorem}

\section{Experimental Evaluation}\label{experiment}

\subsection{Experimental Setting}
\label{Exp_set}

We conduct experiments on the standard testbed for ATGs. More specifically, we use three different multi-player parametric versions of games: \emph{Kuhn poker} \cite{farina2018ex, kuhn1950simplified}, \emph{Leduc poker} \cite{farina2018ex, DBLP:conf/uai/SoutheyBLPBBR05} and \emph{Goofspiel} \cite{farina2021connecting, ross1971goofspiel}, as they are commonly used for experimental evaluation \cite{farina2021connecting}. Specifically, unlike the other two scenarios, \emph{Goofspiel} involves changes in the amount of players' private information during the game. The number of players in these games is parameterized for flexibility. The specific rules for these games are provided in Appendix C.

We denote the number of opponents by $m$ and the number of team members by $n$. For brevity, we use the following symbols to describe the parameters of the experiments:
\begin{itemize}
    \item \textbf{$mn$K$r$}: \emph{Kuhn poker} with $r$ ranks;
    \item \textbf{$mn$L$rc$}: \emph{Leduc poker} with $r$ ranks and $c$ indistinguishable suits. The default maximum number of bets allowed per betting round is $1$;
    \item \textbf{$mn$G}: \emph{Goofspiel} with three ranks.
\end{itemize}

In this work, we adopt the state-of-the-art method that can be combined with 2-player game algorithms as the baseline. The baseline is TPICA proposed by \citet{carminati2022marriage}. Since TPICA is not open-source, we reproduce it based on the information provided in their paper. To ensure a fair comparison, we use the \emph{counterfactual regret minimization plus}, an effective algorithm for finding NE in 2p0s games, in both our method and the baseline. All experiments are run on a machine with 18-core 2.7GHz CPU and 250GB memory.

\subsection{Experimental Results}
\label{Exp_res}
In Table \ref{tab_res}, we use the number of total nodes to represent the size of different games, where the column `Original' represents the scale of the original ATG, and the other columns represent the game size after being transformed by TPICA and our method, respectively. Furthermore, we also provide specific data for the coordinator and adversary nodes. The running time is provided in the column `Runtime'. The TPICA and MPTA algorithms are run for comparison under the same machine configuration and identical experimental conditions. In the four cases of 12\textbf{K}3, 12\textbf{K}4, 12\textbf{K}6 and 12\textbf{L}33 where both MPTA and TPICA are applicable, the total time required by our approach to compute a TMECor is $0.76s$, $9.26s$, $144s$, and $240s$ respectively, which are $182.89 \times$, $168.47 \times$, $694.44 \times$ and $233.98 \times$ faster than TPICA. These results show that our method effectively reduces the game size, significantly improving solving speed by several orders of magnitude. It is worth noting that, in certain large-scale game scenarios where TPICA is unable to transform the original game tree, MPTA can still effectively support the computation of TMECor. In particular, 14\textbf{K}6 and 14\textbf{L}33 are 5-player cases that have never been used as experiments by previous algorithms due to their sheer size. We also observe by the node data in Table \ref{tab_res} that the reason for the speed-up is mainly due to the PIPB structure, which significantly reduces the number of adversary nodes and temporary chance nodes. Furthermore, we conduct detailed analyses on the solving efficiency and execution efficiency of the algorithms.

\begin{figure}
\centering
\subfigure[12\textbf{K}3]{
\label{res1_1}
\includegraphics[width=0.22\textwidth]{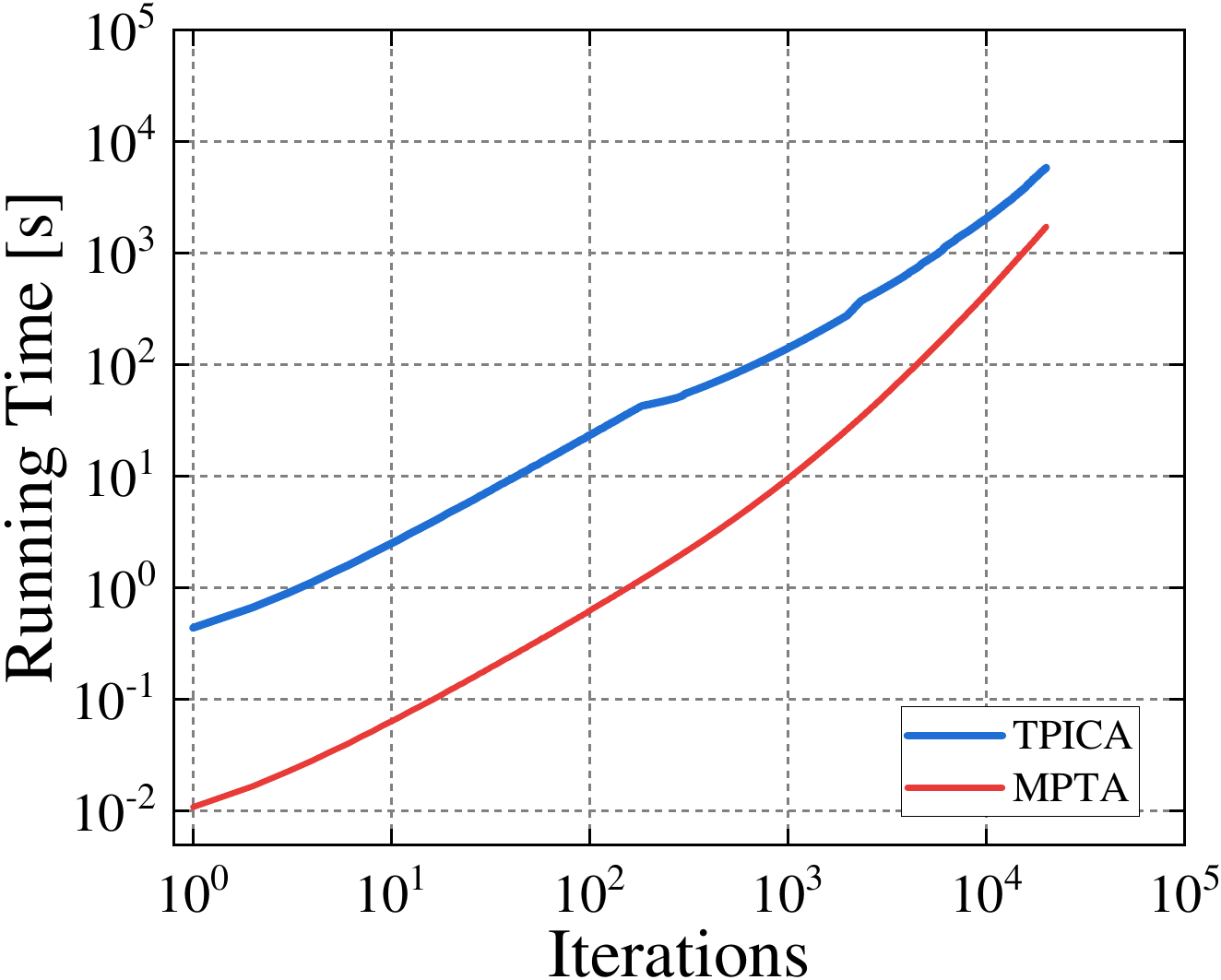}}
\subfigure[12\textbf{K}4]{
\label{res1_2}
\includegraphics[width=0.22\textwidth]{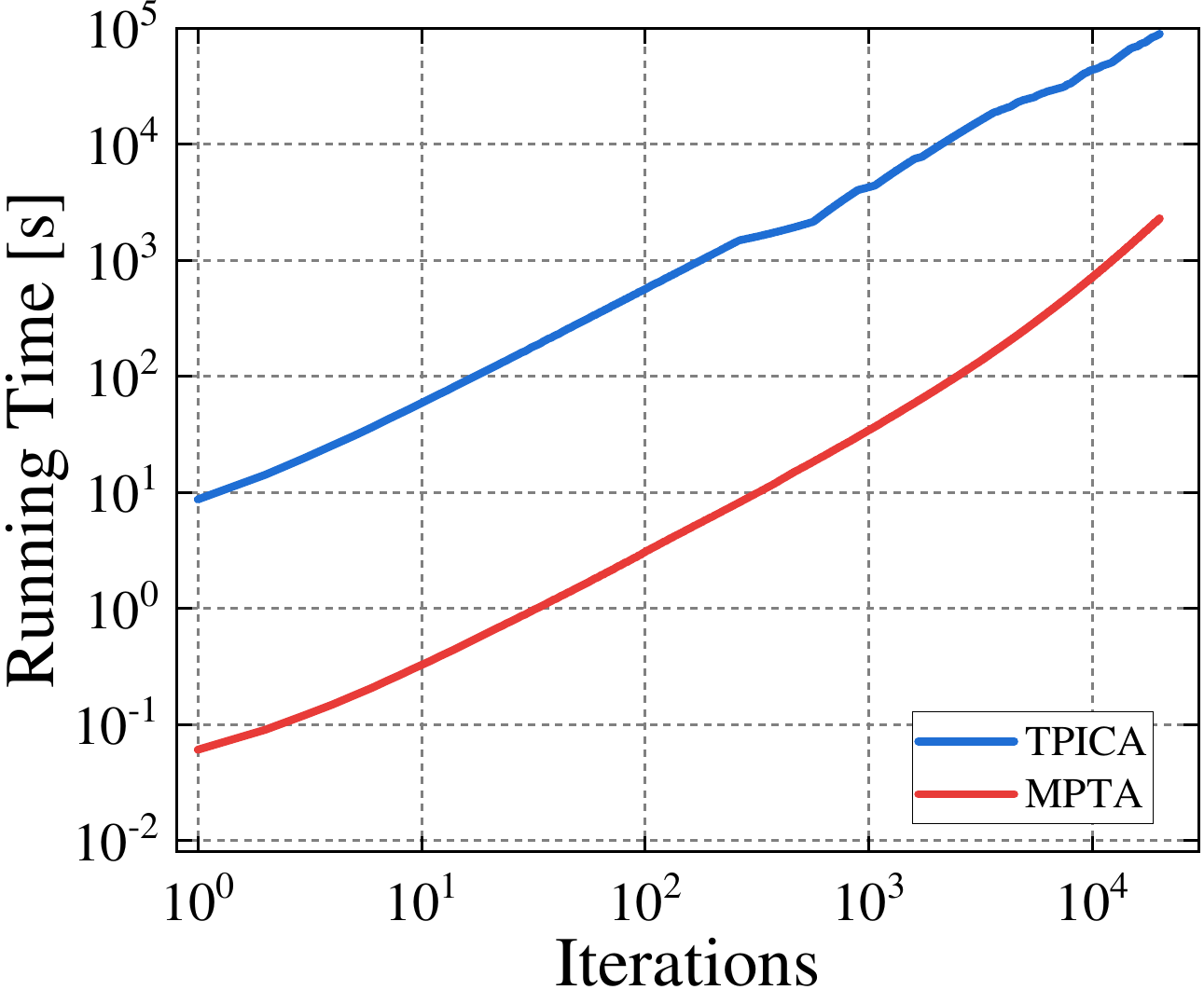}}
\subfigure[12\textbf{K}6]{
\label{res1_3}
\includegraphics[width=0.22\textwidth]{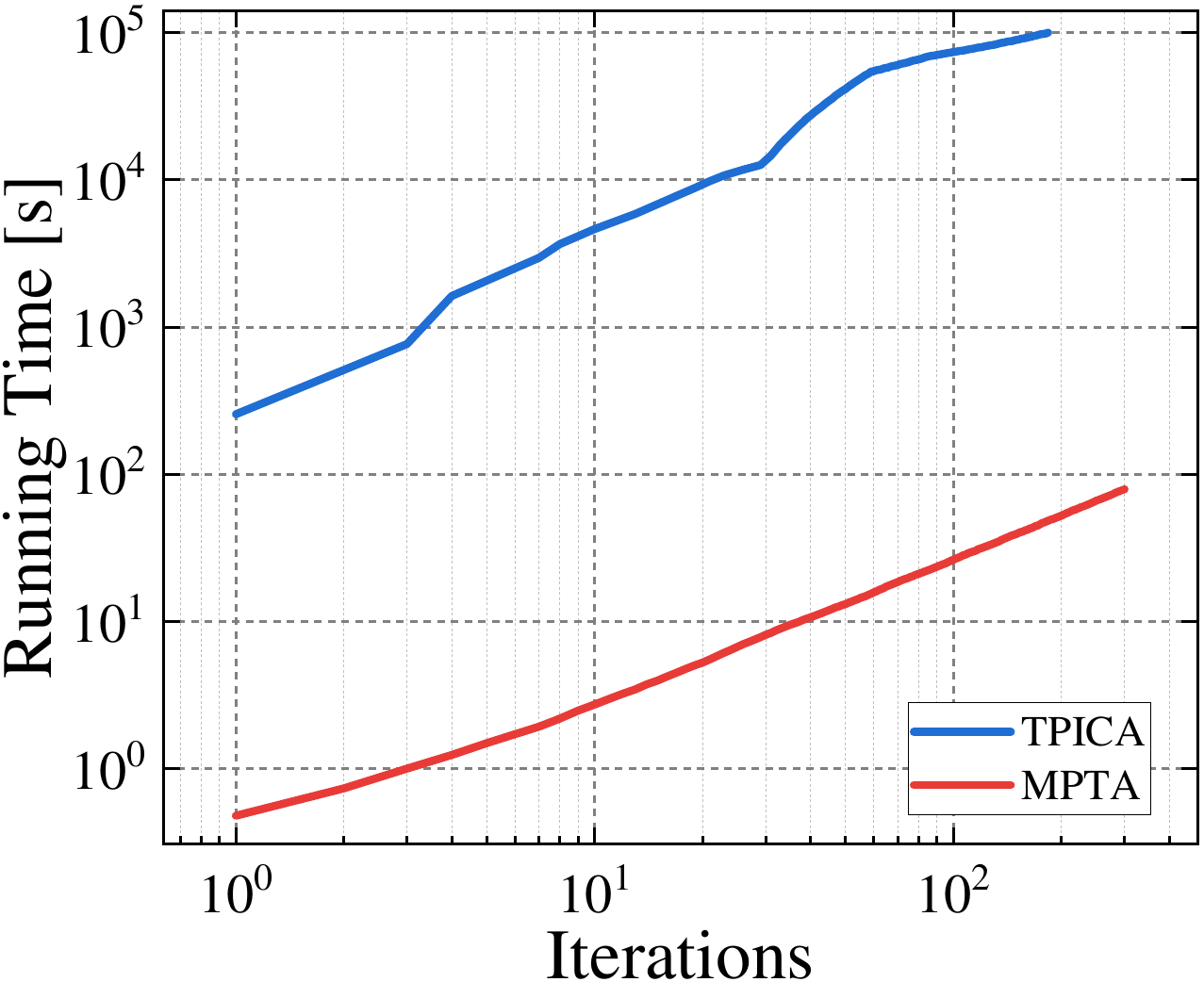}}
\subfigure[12\textbf{L}33]{
\label{res1_4}
\includegraphics[width=0.22\textwidth]{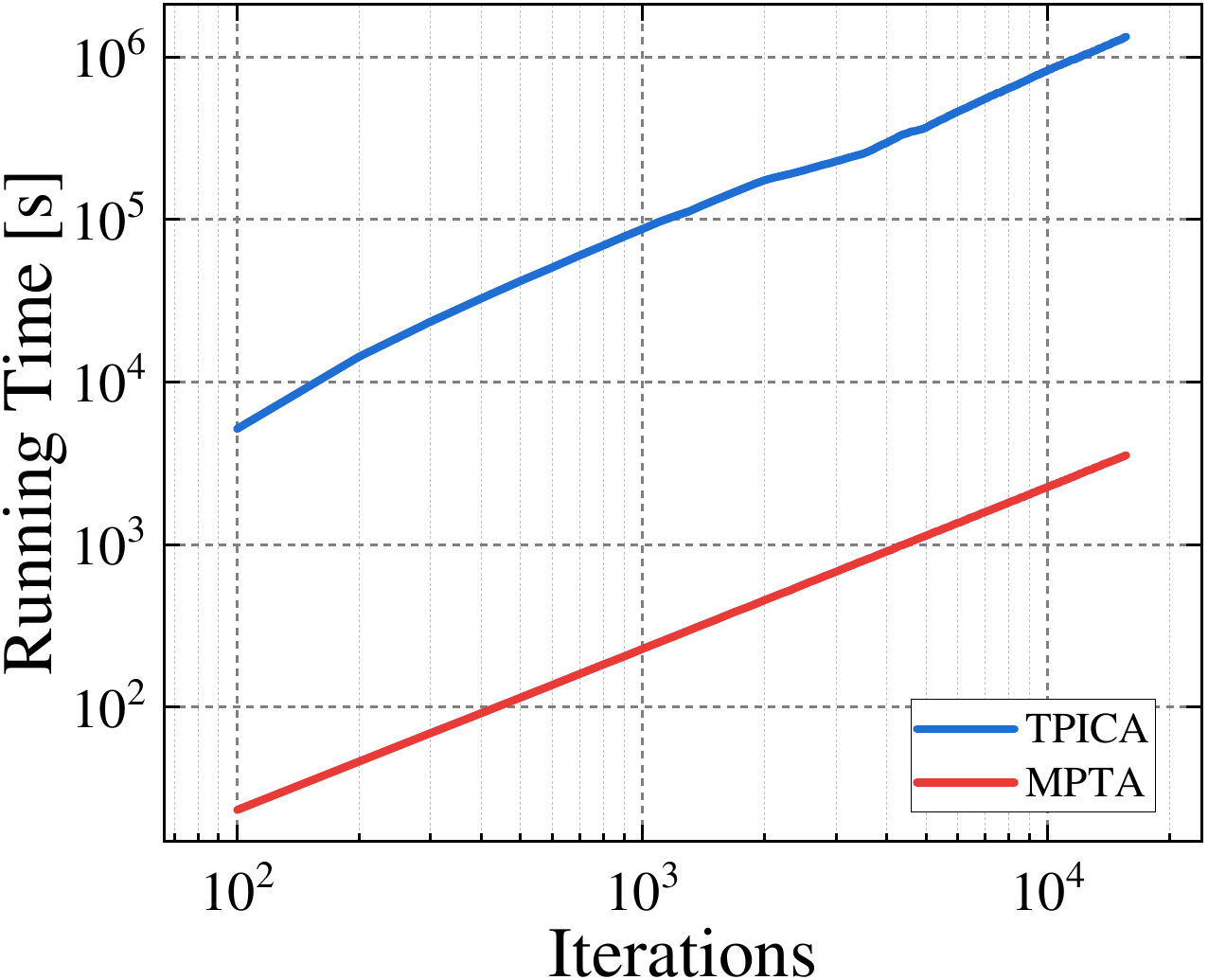}}
\caption{Comparison of runtime within the same number of iterations. All schemes except for 12\textbf{K}6 have been iterated for 20,000 rounds, as the TPICA is too time-consuming to run more rounds.}  
\label{fig:it_time}
\end{figure}

\begin{figure*}
\centering
\subfigure[12\textbf{K}3]{
\label{res2_1}
\includegraphics[width=0.25\textwidth]{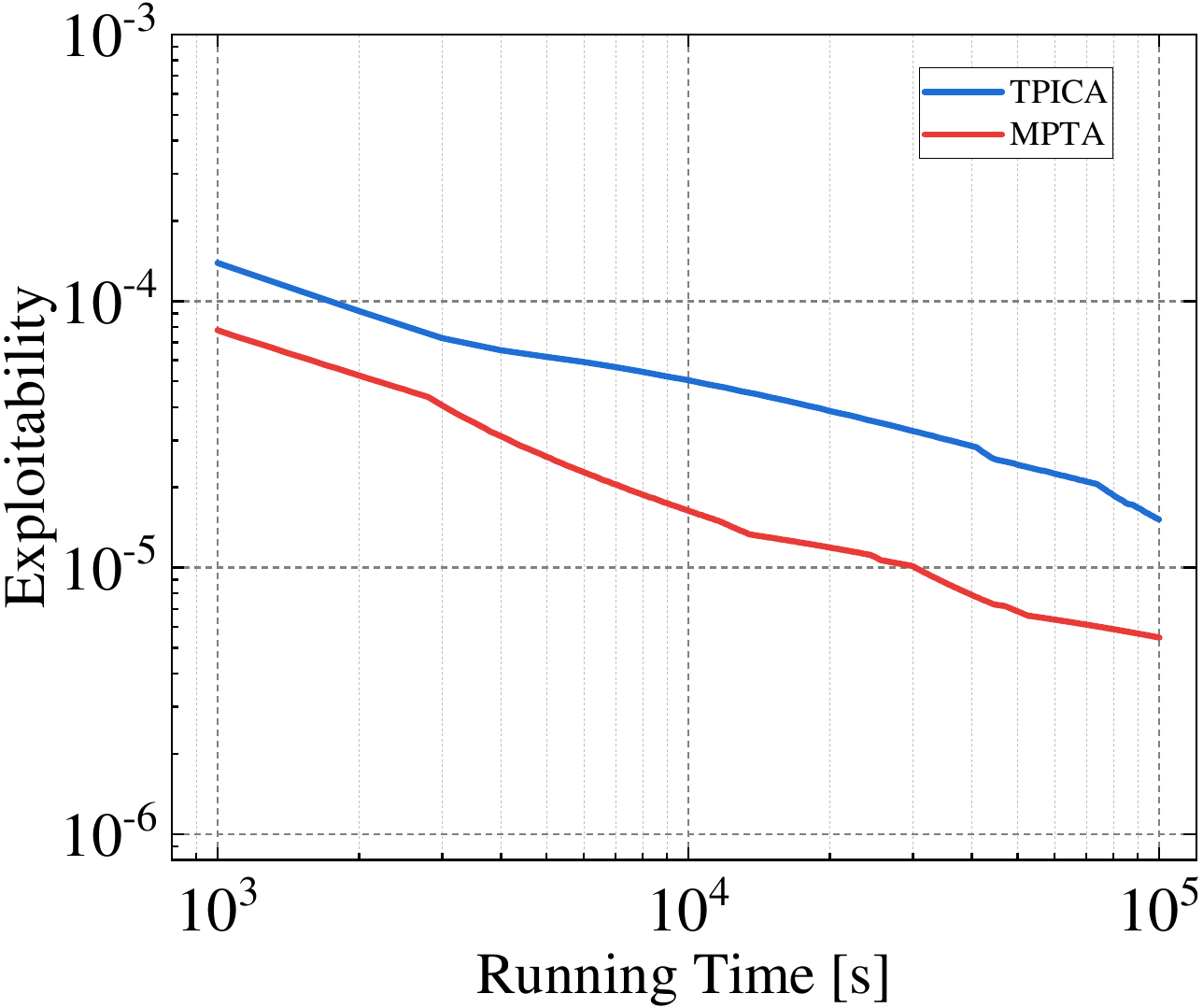}}
\subfigure[12\textbf{L}33]{
\label{res2_2}
\includegraphics[width=0.25\textwidth]{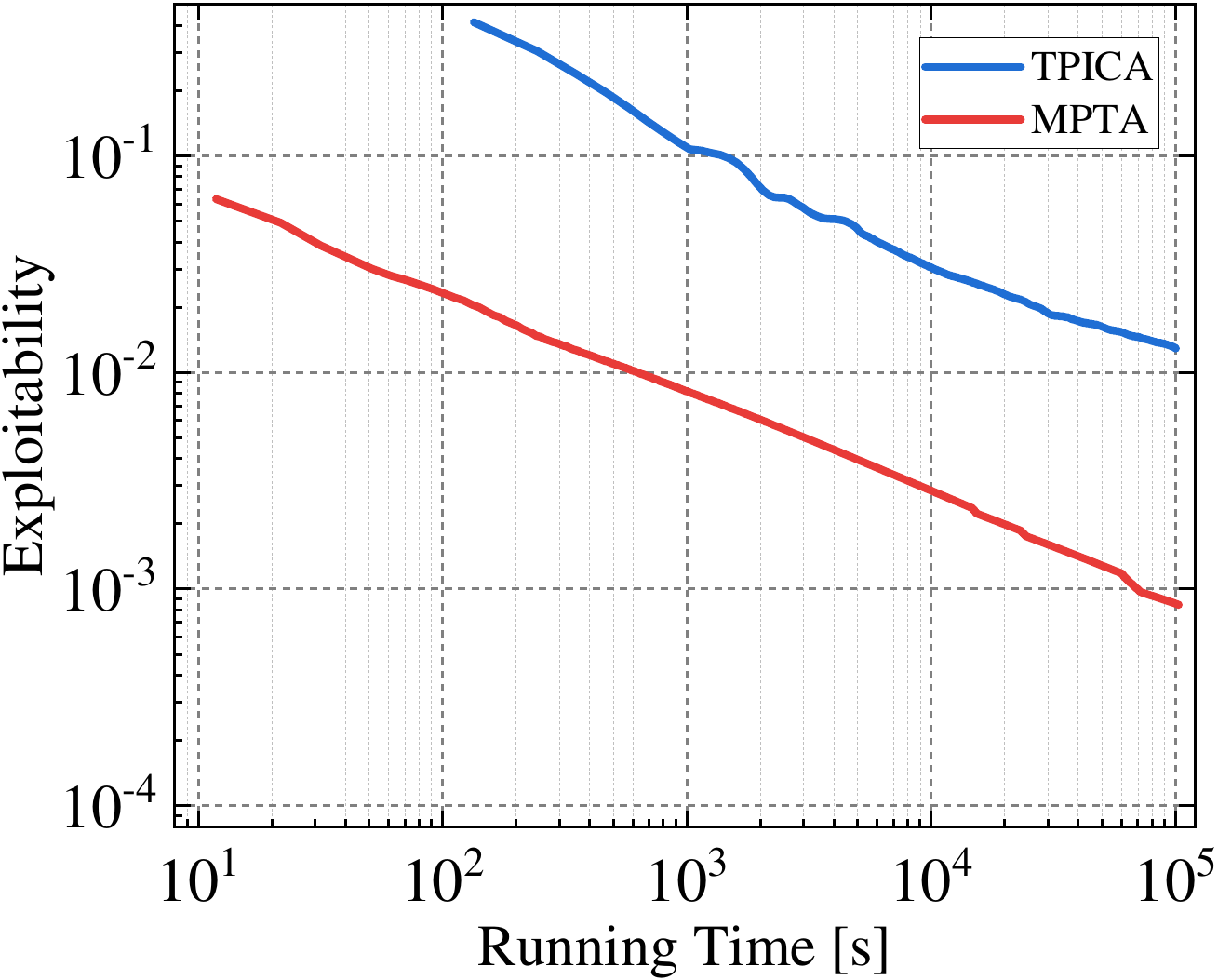}}
\subfigure[12\textbf{K}6]{
\label{res2_3}
\includegraphics[width=0.25\textwidth]{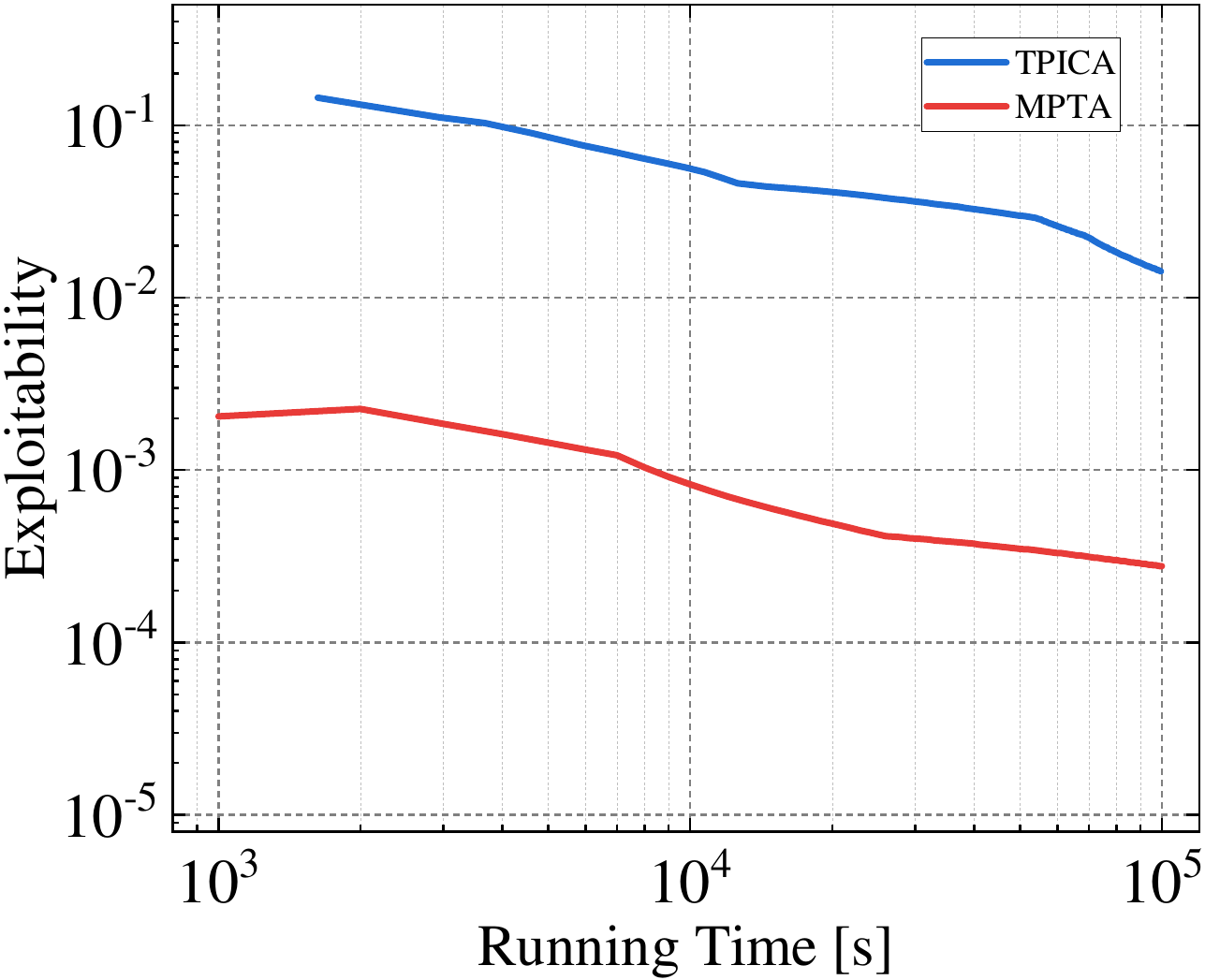}}
\subfigure[14\textbf{K}6]{
\label{res2_4}
\includegraphics[width=0.25\textwidth]{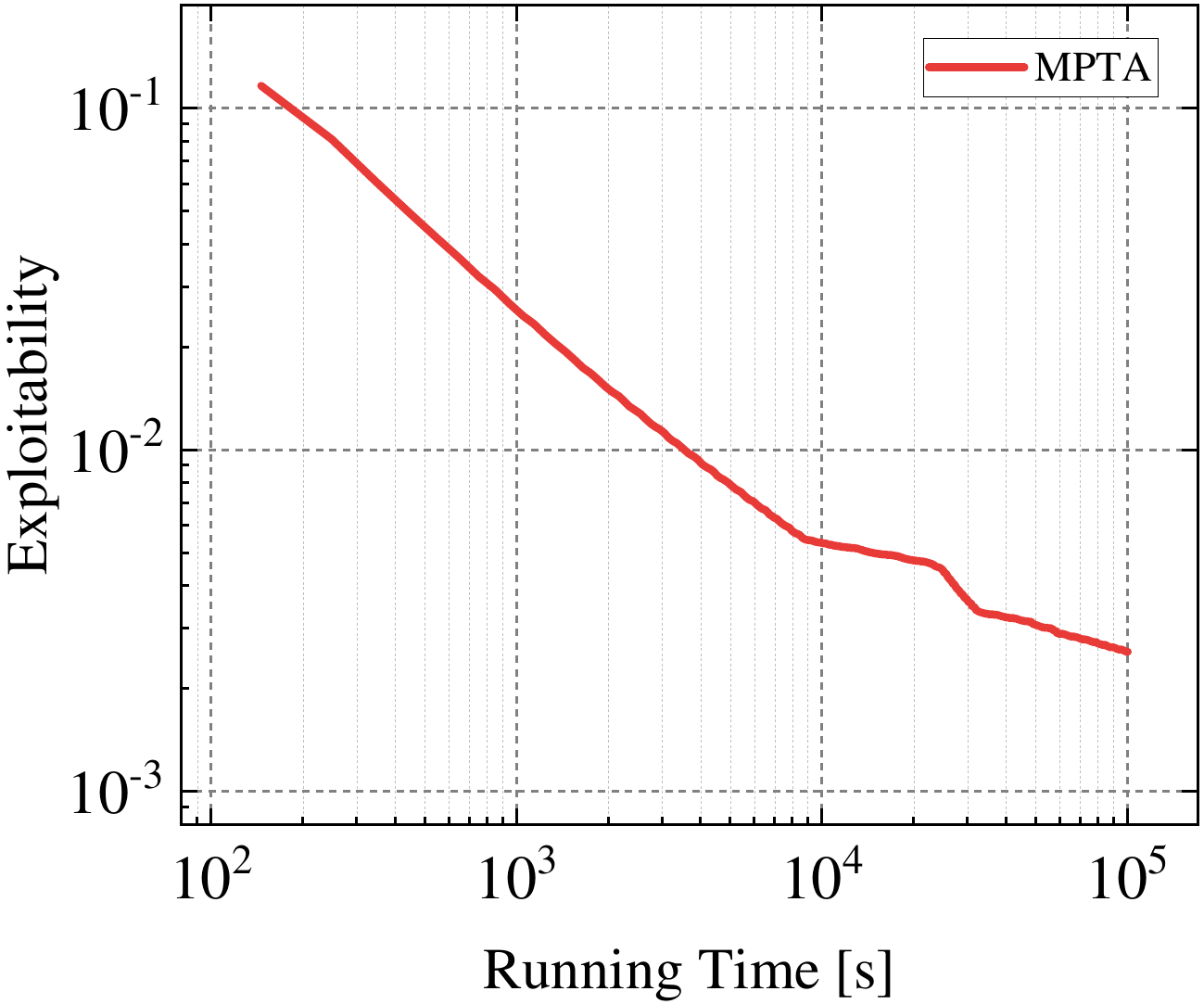}}
\subfigure[14\textbf{L}33]{
\label{res2_5}
\includegraphics[width=0.25\textwidth]{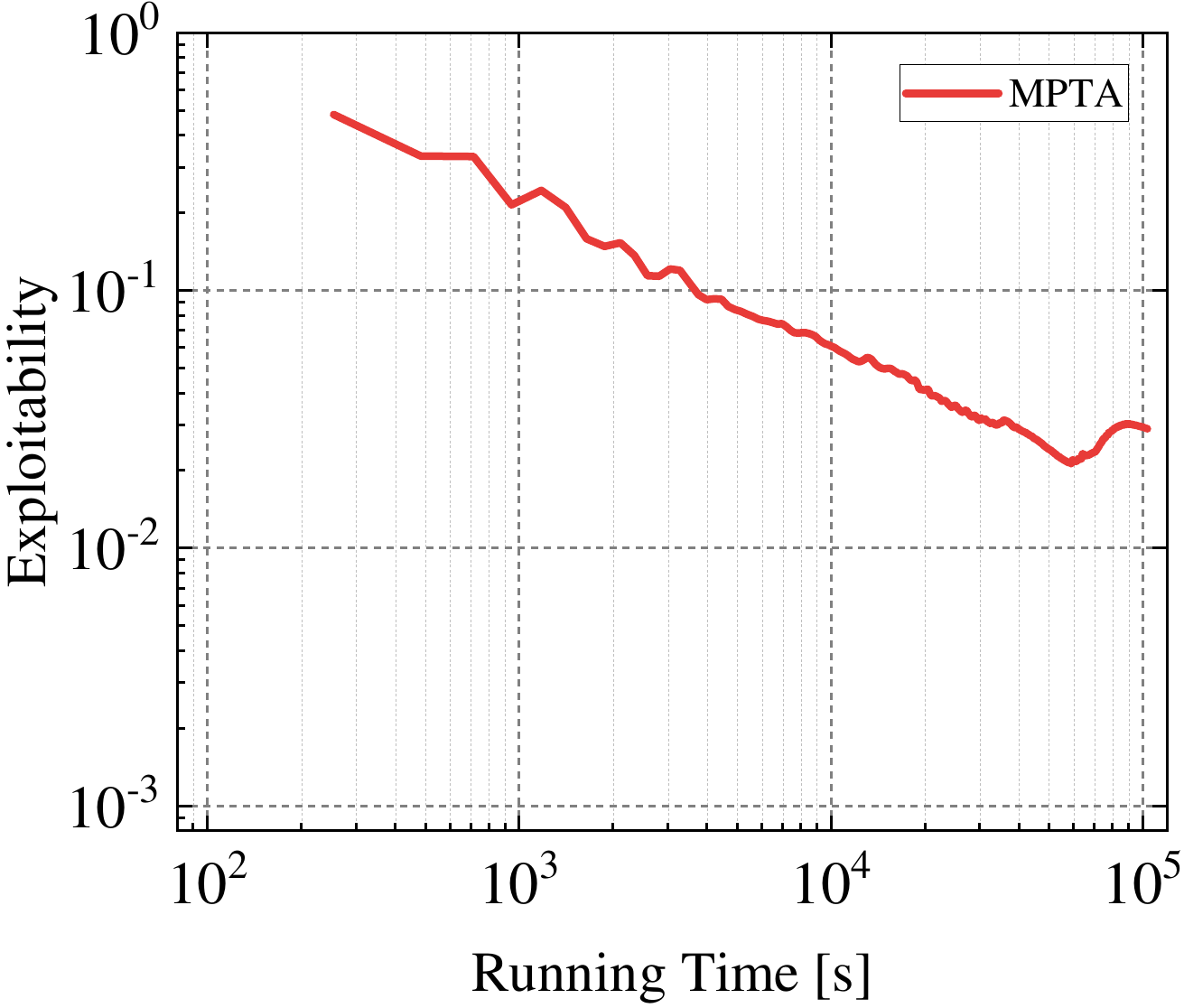}}
\subfigure[12\textbf{G} and 13\textbf{G}]{
\label{res2_6}
\includegraphics[width=0.25\textwidth]{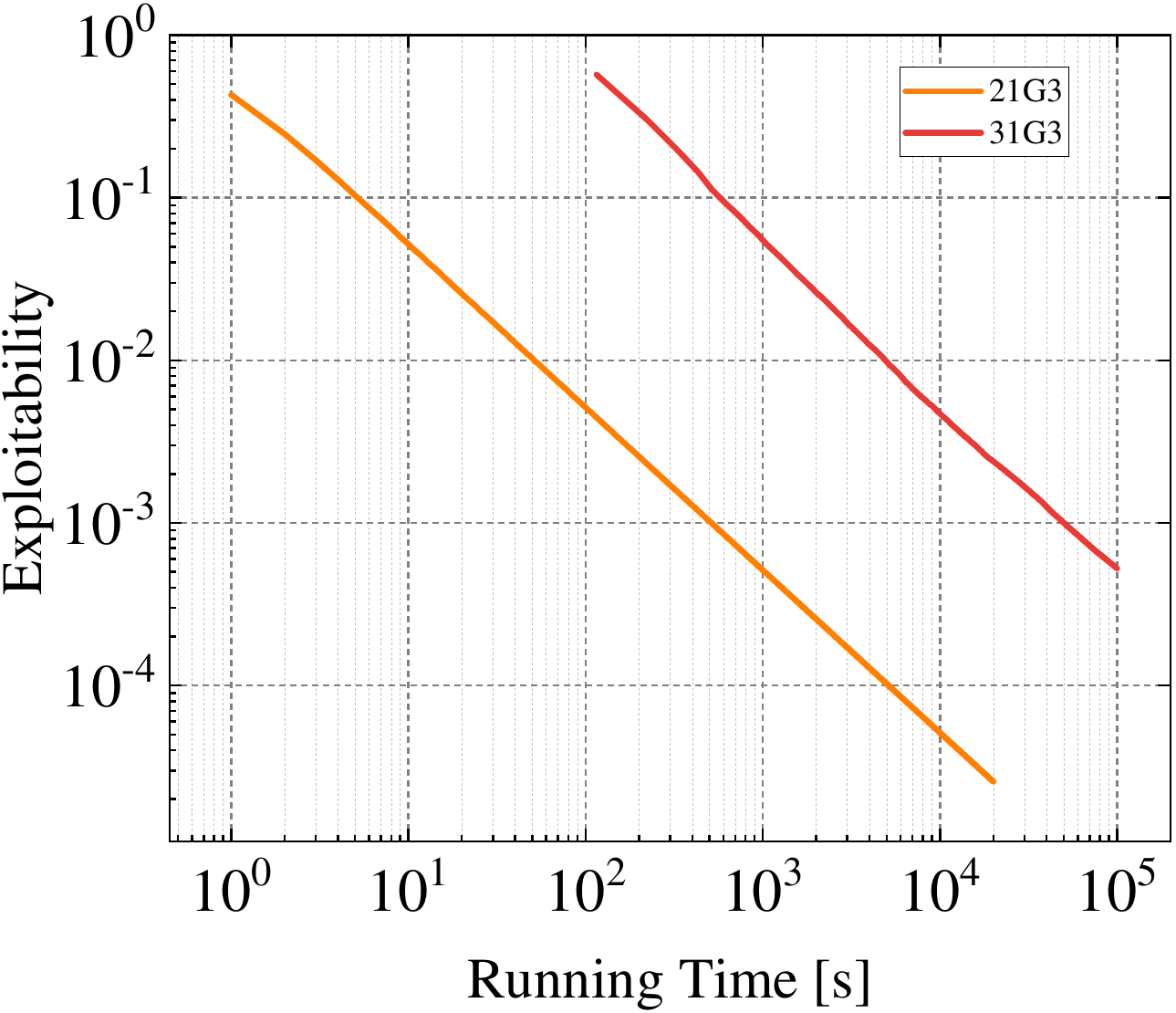}}
\caption{Comparison of exploitability in the same running time. All experiments except 21G run for 100,000 seconds. TPICA fails to work due to out-of-memory in 14\textbf{K}6 and 14\textbf{L}33 and cannot run on Goofspiel due to changes in private information.}  
\label{fig:time_expl}
\end{figure*}
\paragraph{Execution Efficiency.}
The process of finding a TMECor is iterative. To evaluate algorithmic execution efficiency, we conducted comparative experiments on the time taken by the algorithms over the same iteration rounds in four distinct scenarios (i.e., 12\textbf{K}3, 12\textbf{K}4, 12\textbf{K}6, and 12\textbf{L}33). As illustrated in Figure \ref{fig:it_time}, MPTA consistently requires less time to compute approximate equilibrium strategy profiles compared to TPICA. Specifically, Figures \ref{res1_1} and \ref{res1_2} show that MPTA takes 1,740 seconds and 2,293 seconds to complete 20,000 iterations in 12\textbf{K}3 and 12\textbf{K}4, while TPICA takes 5,878 seconds and 89,337 seconds under the same conditions, representing speed improvements of 182.89 and 168.47 times, respectively. Notably, the scales of 12\textbf{K}6 and 12\textbf{L}33 are significantly larger. Figures \ref{res1_3} and \ref{res1_4} show that in these scenarios, MPTA's advantages are even more prominent. In 12\textbf{K}6, MPTA takes 79 seconds for 300 rounds, but TPICA takes 99,600 seconds for 184 rounds. In 12\textbf{L}33, MPTA and TPICA take 3,555 seconds and 1,327,652 seconds for 15,600 iterations, respectively, showing speed improvements of 694.44 and 233.98 times. These results highlight the notable enhancement in the execution efficiency by our method. 
\paragraph{Solving Efficiency.}
A smaller exploitability indicates that the current strategy profile is closer to the TMECor. To compare solving efficiency, we tested the change in exploitability over time within a limited runtime of $10^5$ seconds on seven game instances from \emph{Kuhn poker}, \emph{Leduc poker} and \emph{Goofspiel}, as shown in Figure \ref{fig:time_expl}. Figure \ref{res2_1}, \ref{res2_2}, and \ref{res2_3} show that MPTA consistently outperforms TPICA. This demonstrates our method's higher computational accuracy within the same runtime. As the game scale increases, the gap between the performance of MPTA and TPICA widens, indicating MPTA's superiority in handling large-scale scenarios 14\textbf{K}6 and 14\textbf{L}33, TPICA fails due to out-of-memory, as shown in Figures \ref{res2_4} and \ref{res2_5}. Although MPTA has not yet converged to an approximate equilibrium within the $10^5$ seconds, it can do so with sufficient runtime. Figure \ref{res2_6} shows MPTA's robust performance on \emph{Goofspiel}, which involves changes in players' private information during the game. TPICA, using fixed \emph{prescriptions} to specify an action for each infoset, cannot handle such games, highlighting our approach's high generalizability.

\section{Conclusion and Discussion}\label{conclusion}
In this paper, we propose a multi-player transformation algorithm that establishes a connection between 2p0s games and adversarial team games. Our method restricts the exponential growth of the transformed game action space by utilizing the PIPB structure. It can handle situations where private information changes during the game. Furthermore, we prove the equilibrium equivalence between the original and transformed game, which provides a theoretical guarantee for our work. We conducted 14 experiments on multiple standard testbeds, all of which demonstrated exceptional performance, further showcasing the effectiveness of our method.

Our method may also be applicable in real-world scenarios. For instance, in environmental protection activities, ATG models team members' inability to communicate while protecting the environment in different regions. In competitive games, team members may be unable to communicate during the game due to the rules. In future work, we plan to introduce the idea of equilibrium refinement, using the strategies obtained in a perfect information environment to guide team members in making decisions in ATGs.


\bibliographystyle{plainnat}
\renewcommand{\bibsection}{\subsubsection*{References}}
\newpage

\appendix
\section{Appendix A}
\subsection{The Proof of Theorem 1}
\begin{theorem}
    Given an ATG $G$ with visibility that satisfies the public-turn-taking property, and its transformed game $G^{\prime}=\emph{MPTA}(G)$. The size of any episode in $G^{\prime}$ is $\mathcal{O}\big((\frac{(\lvert \Omega\rvert-1)!}{(\lvert \Omega\rvert-\lvert \mathcal{T}\rvert)!}\lvert A\rvert)^{\lvert \mathcal{T}\rvert}\big)$.
\end{theorem}
\begin{proof}
    We assume that each player has $\lvert A\rvert$ available actions at every state in the original game.
    During the traversal of the original game tree, the dummy player nodes will provide all possible private information from the teammates. Therefore, the number of available actions at the dummy player nodes is $\frac{(\lvert \Omega\rvert-1)!}{(\lvert \Omega\rvert-\lvert \mathcal{T}\rvert)!}$. The coordinator inherits the team members' actions since they are \emph{public}. When the team in the game consists of two players, the size of any episode in $G^{\prime}$ is given by:
    \begin{equation*}
    \begin{aligned}
        &\frac{(\lvert \Omega\rvert-1)!}{(\lvert \Omega\rvert-\lvert \mathcal{T}\rvert)!} +\frac{(\lvert \Omega\rvert-1)!}{(\lvert \Omega\rvert-\lvert \mathcal{T}\rvert)!}\lvert A\rvert+(\frac{(\lvert \Omega\rvert-1)!}{(\lvert \Omega\rvert-\lvert \mathcal{T}\rvert)!})^2\lvert A\rvert \\ &+ (\frac{(\lvert \Omega\rvert-1)!}{(\lvert \Omega\rvert-\lvert \mathcal{T}\rvert)!})^2\lvert A\rvert^{2}
    \end{aligned}
    \end{equation*}
    When the team in the game consists of three players, the size of any episode in $G^{\prime}$ is given by:
    \begin{equation*}
    \begin{aligned}
        &\frac{(\lvert \Omega\rvert-1)!}{(\lvert \Omega\rvert-\lvert \mathcal{T}\rvert)!} +\frac{(\lvert \Omega\rvert-1)!}{(\lvert \Omega\rvert-\lvert \mathcal{T}\rvert)!}\lvert A\rvert+(\frac{(\lvert \Omega\rvert-1)!}{(\lvert \Omega\rvert-\lvert \mathcal{T}\rvert)!})^{2}\lvert A\rvert \\ &+ (\frac{(\lvert \Omega\rvert-1)!}{(\lvert \Omega\rvert-\lvert \mathcal{T}\rvert)!})^2\lvert A\rvert^{2}+(\frac{(\lvert \Omega\rvert-1)!}{(\lvert \Omega\rvert-\lvert \mathcal{T}\rvert)!})^{3}\lvert A\rvert^{2} \\ &+ (\frac{(\lvert \Omega\rvert-1)!}{(\lvert \Omega\rvert-\lvert \mathcal{T}\rvert)!})^{3}\lvert A\rvert^{3}
    \end{aligned}
    \end{equation*}
    Thus, extending to the general case where the team consists of $\lvert \mathcal{T}\rvert$ players, the size of any episode in $G^{\prime}$ is given by:
    \begin{equation*}
        \sum_{n=1}^{\lvert \mathcal{T}\rvert}\big(\frac{(\lvert \Omega\rvert-1)!}{(\lvert \Omega\rvert-\lvert \mathcal{T}\rvert)!}\big)^{n}(\lvert A\rvert^{n-1}+\lvert A\rvert^{n}).
    \end{equation*}
    Let $S_{1}=\sum_{n=1}^{\lvert \mathcal{T}\rvert}(\frac{(\lvert \Omega\rvert-1)!}{(\lvert \Omega\rvert-\lvert \mathcal{T}\rvert)!})^{n}\lvert A\rvert^{n-1}$, and $S_{2}=\sum_{n=1}^{\lvert \mathcal{T}\rvert}(\frac{(\lvert \Omega\rvert-1)!}{(\lvert \Omega\rvert-\lvert \mathcal{T}\rvert)!})^{n}\lvert A\rvert^{n}$. Then, we have:
    \begin{equation*}
        \sum_{n=1}^{\lvert \mathcal{T}\rvert}\big(\frac{(\lvert \Omega\rvert-1)!}{(\lvert \Omega\rvert-\lvert \mathcal{T}\rvert)!}\big)^{n}(\lvert A\rvert^{n-1}+\lvert A\rvert^{n})=S_{1}+S_{2}.
    \end{equation*}
    First, consider $S_{1}$:
    \begin{equation*}
    \begin{aligned}
        S_{1} = &\frac{(\lvert \Omega\rvert-1)!}{(\lvert \Omega\rvert-\lvert \mathcal{T}\rvert)!}+\big(\frac{(\lvert \Omega\rvert-1)!}{(\lvert \Omega\rvert-\lvert \mathcal{T}\rvert)!}\big)^{2}\lvert A\rvert+\dots \\ &+\big(\frac{(\lvert \Omega\rvert-1)!}{(\lvert \Omega\rvert-\lvert \mathcal{T}\rvert)!}\big)^{\lvert \mathcal{T}\rvert}\lvert A\rvert^{\lvert \mathcal{T}\rvert-1}.
        \end{aligned}
    \end{equation*}
    $S_{1}$ meets the criteria for a finite geometric series. Using the geometric series sum formula, we have:
    \begin{equation*}
        S_{1}=\frac{(\lvert A\rvert\frac{(\lvert \Omega\rvert-1)!}{(\lvert \Omega\rvert-\lvert \mathcal{T}\rvert)!})^{\lvert \mathcal{T}\rvert}-1}{\lvert A\rvert\frac{(\lvert \Omega\rvert-1)!}{(\lvert \Omega\rvert-\lvert \mathcal{T}\rvert)!}-1}\frac{(\lvert \Omega\rvert-1)!}{(\lvert \Omega\rvert-\lvert \mathcal{T}\rvert)!}
    \end{equation*}
    Similarly, consider $S_{2}$:
    \begin{equation*}
    \begin{aligned}
        S_{2}= &\frac{(\lvert \Omega\rvert-1)!}{(\lvert \Omega\rvert-\lvert \mathcal{T}\rvert)!}\lvert A\rvert + \big(\frac{(\lvert \Omega\rvert-1)!}{(\lvert \Omega\rvert-\lvert \mathcal{T}\rvert)!}\big)^{2}\lvert A\rvert^{2}+\dots \\ & +\big(\frac{(\lvert \Omega\rvert-1)!}{(\lvert \Omega\rvert-\lvert \mathcal{T}\rvert)!}\big)^{\lvert \mathcal{T}\rvert} +\lvert A\rvert^{\lvert \mathcal{T}\rvert}.
        \end{aligned}
    \end{equation*}
    Thus, 
    \begin{equation*}
        S_{2}=\frac{(\lvert A\rvert\frac{(\lvert \Omega\rvert-1)!}{(\lvert \Omega\rvert-\lvert \mathcal{T}\rvert)!})^{\lvert \mathcal{T}\rvert}-1}{\lvert A\rvert\frac{(\lvert \Omega\rvert-1)!}{(\lvert \Omega\rvert-\lvert \mathcal{T}\rvert)!}-1} \frac{(\lvert \Omega\rvert-1)!}{(\lvert \Omega\rvert-\lvert \mathcal{T}\rvert)!}\lvert A\rvert.
    \end{equation*}
    Adding $S_{1}$ and $S_{2}$, we have:
    \begin{equation*}
        S_{1}+S_{2}=(1+\lvert A\rvert)\frac{(\lvert \Omega\rvert-1)!}{(\lvert \Omega\rvert-\lvert \mathcal{T}\rvert)!}\frac{(\lvert A\rvert\frac{(\lvert \Omega\rvert-1)!}{(\lvert \Omega\rvert-\lvert \mathcal{T}\rvert)!})^{\lvert \mathcal{T}\rvert}-1}{\lvert A\rvert\frac{(\lvert \Omega\rvert-1)!}{(\lvert \Omega\rvert-\lvert \mathcal{T}\rvert)!}-1}.
    \end{equation*}
    Therefore, the size of any episode in $G^{\prime}$ is $\mathcal{O}\big((\frac{(\lvert \Omega\rvert-1)!}{(\lvert \Omega\rvert-\lvert \mathcal{T}\rvert)!}\lvert A\rvert)^{\lvert \mathcal{T}\rvert}\big)$.
    
    This concludes the proof.
\end{proof}

Let the game transformed by TPICA be denoted as $\hat{G}$. In $\hat{G}$, each team member node, except for the team member node who first acts, corresponds to an additional dummy player node. We recognize that, according to the \emph{prescription} property in TPICA, every coordinator node except for the first one will be matched with a dummy player node. By applying a derivation process similar to Theorem 1, the size of any episode in $\hat{G}$ is $2\sum_{n=1}^{\lvert \mathcal{T}\rvert}(\lvert A\rvert^{\lvert \Omega\rvert})^{n}$. Using the geometric series sum formula, the above equation is equal to $2\lvert A\rvert^{\lvert \Omega\rvert}\frac{(\lvert A\rvert^{\lvert \Omega\rvert})^{\lvert \mathcal{T}\rvert}-1}{\lvert A\rvert^{\lvert \Omega\rvert}-1}$ (i.e., $\mathcal{O}\big((\lvert A\rvert^{\lvert \Omega\rvert})^{\lvert \mathcal{T}\rvert}\big)$). Then, we compare the bases of the two results (i.e., $\frac{(\lvert \Omega\rvert-1)!}{(\lvert \Omega\rvert-\lvert \mathcal{T}\rvert)!}\lvert A\rvert$ and $\lvert A\rvert^{\lvert \Omega\rvert}$). Clearly, the exponential growth rate of the latter far exceeds the polynomial growth rate of the former.

\subsection{The Proof of Lemma 1}
\begin{lemma}
    Given an ATG $G$ with visibility that satisfies the public-turn-taking property, and the transformed game $G^{\prime}=\emph{MPTA}(G)$. For any joint reduced pure strategy $\pi_{\mathcal{T}}$ in $G$, it can be mapped to a corresponding strategy $\pi_{t}$ in $G^{\prime}$, and vice versa.
\end{lemma}

\begin{proof}
    We can prove Lemma~\ref{lemma:strategy} by recursively traversing $G$ and $G^{\prime}$ in a depth-first pre-order manner.
    \paragraph{Case 1: from $\pi_{\mathcal{T}}$ to $\pi_{t}$.} Let $h$ and $h^{\prime}$ denote the current nodes reached in $G$ and $G^{\prime}$, respectively. $G$ satisfies the public-turn-taking, ensuring $h$ and $h^{\prime}$ either represent the same player or are both terminal nodes. Initializing $h$ with the chance player node. $\Omega$ is a set of all private information in $G$. When constructing $\pi_{t}$:

    \begin{itemize}
        \item[1)] For the chance node: $\pi_{c}=\pi_{c^{\prime}}$ always holds as our algorithm does not modify the chance node. This ensures that the actions specified by the chance player's strategy in the original game are the same as those in the transformed game.
        
        \item[2)] For opponent nodes: The proof is identical to 1).
        
        \item[3)] For team member nodes: Let $\pi_{\mathcal{T}}[I(h)]$ represent the joint reduced pure strategy of $\pi_{\mathcal{T}}$ at infoset $I(h)$. During the traversal, our algorithm expands $h$ based on all the possibly private information of teammates into $\lvert \Omega\rvert-1$ nodes. These nodes belong to the same infoset, denoted as $I^{\prime}$. Let $\pi_{t}[I^{\prime}]$ represent the reduced pure strategy of $\pi_{t}$ at $I^{\prime}$, where $I^{\prime}\in S_{t}(h)$. When $t$ is in a public state, $\pi_{\mathcal{T}}[I(h)]=\pi_{t}(I^{\prime})$.

        \item[4)] For terminal nodes: When reaching the terminal node through the above process, our algorithm ensures the following holds: $u_t^{\prime}=\sum_{i\in \mathcal{T}}u_i(h)$ and $u_{o}^{\prime}(h^{\prime}) =u_o(h)=-\sum_{i\in \mathcal{T}}u_i(h)$.
    \end{itemize}
    \paragraph{Case 2: from $\pi_{t}$ to $\pi_{\mathcal{T}}$.} The proof follows the same points as the previous case.
\end{proof}

\subsection{The Proof of Theorem 2}
\begin{theorem}
    Given a public-turn-taking ATG $G$ with visibility, and its transformed game $G^{\prime}=\emph{MPTA}(G)$, they have equivalent payoffs.
\end{theorem}

\begin{proof}
    The proof directly relies on Lemma~\ref{lemma:strategy}. Specifically, for any strategy $\pi_{\mathcal{T}}$ in $G$, we can find a corresponding strategy $\pi_{t}$ in $G^{\prime}$ that yields the same payoff. Similarly, for any strategy $\pi_{t}$ in $G^{\prime}$, there exists a payoff-equivalent strategy $\pi_{\mathcal{T}}$ in $G$. This ensures that the payoff for the players remains unchanged whether they choose $\pi_{t}$ in $G^{\prime}$ or $\pi_{\mathcal{T}}$ in $G$.
\end{proof}

\subsection{The Proof of Theorem 3}
\begin{theorem}
    Given an ATG $G$ with visibility that satisfies the public-turn-taking property, and its transformed game $G^{\prime}=\emph{MPTA}(G)$. If $\mu_{t}^{*}$ is an NE in $G^{\prime}$, then strategy $\mu_{T}^{*}: \mu_{t}^{*}\mapsto \mu_{\mathcal{T}}^{*}$ is a TMECor in $G$.
\end{theorem}

\begin{proof}
    For brevity, let $u_{t}$ and $u_{\mathcal{T}}$ represent $u_{t}(\pi_{c},\pi_{o},\pi_{t})$ and $u_{\mathcal{T}}(\pi_c, \pi_o,\pi_{\mathcal{T}})$, respectively. If $\mu_{t}^{*}$ is an NE in $G^{\prime}$, then by the definition of NE, the following holds:
    \begin{equation*}
        \mu_{t}^{*} \in \arg\max_{\mu_{t}} \min_{\mu_{o}} \sum_{\substack{\pi_c \in \Pi_c \\ \pi_o\in \Pi_o\\ \pi_t \in \Pi_t}}\mu_{c}(\pi_{c})\mu_{o}(\pi_{o})\mu_{t}(\pi_{t})u_{t}.
    \end{equation*}
    If $\mu_{\mathcal{T}}^{*}$ is a TMECor, it satisfies:
    \begin{equation*}
        \mu_{\mathcal{T}}^{*} \in \arg\max_{\mu_{\mathcal{T}}}\min_{\mu_{o}}\sum_{\substack{\pi_{c}\in \Pi_{c} \\ \pi_{o}\in \Pi_{o} \\ \pi_{\mathcal{T} \in \Pi_{\mathcal{T} }}}} \mu_{c}(\pi_{c}) \mu_{o}(\pi_{o}) \mu_{\mathcal{T}}(\pi_{\mathcal{T}}) u_{\mathcal{T}}.
    \end{equation*}

\begin{figure*}[t]
\centering
\includegraphics[width=1\textwidth]{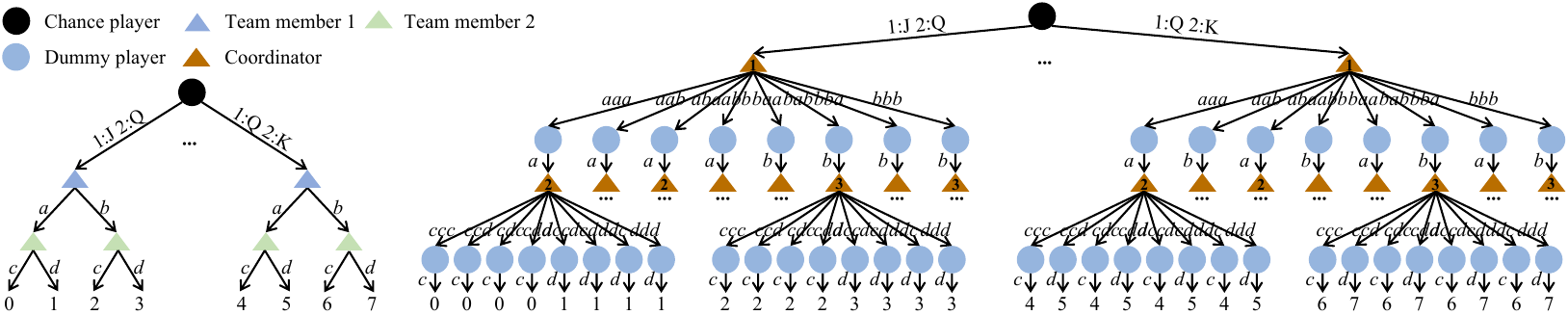}
\caption{Example of game transformation. ``\textbf{\dots}'' indicates omitted branches. The nodes of a player with the same number are in the same infoset. \textbf{Left:} Original ATG omitting the opponent. \textbf{Right:} Result of transforming the game on the left using TPICA.}
\label{fig:TransformedTPICA}
\end{figure*}
    
    Let $\min_{TMECor}(\mu_{\mathcal{T}})$ and $\min_{NE}(\mu_t)$ denote the inner minimization problems for TMECor and NE, respectively.

    Assume there exists a strategy $\mu_{\mathcal{T}}^{\prime}$ such that its value under the definition of TMECor is greater than that of $\mu_{\mathcal{T}}^{*}$. That is, $\min_{TMECor}(\mu_{\mathcal{T}}^{\prime}) > \min_{TMECor}(\mu_{\mathcal{T}}^{*})$.
    According to Lemma~\ref{lemma:strategy}, there exists a strategy $\mu_{t}^{\prime}$ such that $\mu_{\mathcal{T}}^{\prime}\mapsto \mu_{t}^{\prime}$. From Theorem~\ref{theorem:payoff}, we then have:
    \begin{equation*}
        \min_{TMECor}(\mu_{\mathcal{T}}^{\prime})=\min_{NE}(\mu_{t}^{\prime}) > \min_{NE}(\mu_{t}^{*}).
    \end{equation*}
    This results in a contradiction, as it implies that $\mu_{t}^{*}$ is not an NE in $G^{\prime}$. Therefore, $\mu_{\mathcal{T}}^{*}$ must be a TMECor in $G$.
\end{proof}

\section{Appendix B}
\subsection{Converted Result by TPICA}
Figure~\ref{fig:TransformedTPICA} is the result of the TPICA transformation. Due to the large number of nodes, we only display a portion of the transformed game tree.

\subsection{Conceptual Explanation of Team-Public-Information Representation}
The left side of Figure~\ref{fig:TransformedTPICA} shows the original ATG tree, ignoring the opponent nodes, with all nodes forming the set $H$. Since every player can only observe the cards dealt to himself by the chance player, the actions of the chance player are \emph{private} to the coordinator. The actions of all other players are observable, so their actions are \emph{public} to the coordinator. In the left side of Figure~\ref{fig:TransformedTPICA}, the two nodes belonging to team member 1 are in two infosets due to the different private information at each node. For team member 2, the actions observed under the same hand are different, making each of the four nodes belonging to team member 2 a separate information set as well. The set of private information in this game is $\Omega={J,Q,K}$. All leaf nodes form the set of terminal nodes $Z$.

The right side of Figure~\ref{fig:TransformedTPICA} shows the result of the game transformation using TPICA. The coordinator represents a team consisting of team members 1 and 2. The coordinator's public state is divided only by actions called \emph{public}. Therefore, the two nodes of team member 1 in the original game tree belong to the same public state. Since we omitted the branches of the chance nodes, marked by ``\textbf{\dots}'', the actual possible deals are: $\left[1:J, 2:Q\right]$, $\left[1:J, 2:K\right]$, $\left[1:Q, 2:J\right]$, $\left[1:Q, 2:K\right]$, $\left[1:K, 2:J\right]$, and $\left[1:K,2:Q\right]$. Nodes with the same private information are in the same infoset. That is, there are three infosets for team member 1 at this level, each consisting of two nodes. According to the concept of \emph{prescription}, every \emph{prescription} should select an action from the three infosets to form recommendations. Thus, there are $2^3$ \emph{prescriptions}, i.e., $aaa$, $aab$, $aba$, $abb$, $baa$, $bab$, $bba$, and $bbb$. The dummy player will select a specific action from these recommendations as the coordinator's available action. For instance, in the left subtree, the dummy player chooses the first action from each \emph{prescription}; in the right subtree, the dummy player chooses the last action from each \emph{prescription}. The same process applies to the coordinator when traversing to team member 2.

\section{Appendix C}
\subsection{Game Instances}
In our work, the number of players in each game scenario is parameterized for flexibility. To articulate the rules of each instance clearly, we will illustrate using the 3-player version as an example.

\begin{itemize}
    \item \textbf{The rule of \emph{Kuhn poker}:}
    In 3-player Kuhn poker, there are three players and $k$ possible cards. Players take turns acting in sequence. Before the game starts, each player pays one chip to the pot and is dealt a private card. The game proceeds with the following steps:
\begin{itemize}
    \item [1)] 
    Player 1 can choose to check or bet. If checking, the betting round continues with step 2); otherwise, the betting round proceeds to step 3).

    \item [2)]
    Player 2 can choose to check or bet. It is important to note that if Player 2 chooses to bet, then Player 1 must decide between folding or calling after Player 3's action. If Player 2 also chooses to check, the betting round continues with step 4).

    \item [3)]
    Player 2 can choose to fold or call.

    \item [4)]
    Player 3 chooses to check or bet. When Player 3 checks, the betting round ends; otherwise, Player 1 and Player 2 must decide between folding or calling.

    \item [5)]
    Player 3 chooses to fold or call. The betting round concludes after her decision.
\end{itemize}

We assume that Player 1 is the adversary, while Player 2 and Player 3 are team members. In the event of the opponent's victory, Players 2 and 3 share the loss. If the team wins, Player 2 and Player 3 share the team's rewards. The $n$-player Kuhn poker adopted in our work is an extended version based on the 3-player Kuhn poker.
    \item \textbf{The rule of \emph{Leduc poker}:} 
    In the 3-player version of adversarial team Leduc poker, the deck contains three suits and $k \geq 3$ card ranks. Each player starts by contributing one chip to the pot and receiving a private card. There are two betting rounds in total. After the first betting round, the community card is revealed. Then, players who have not folded proceed to the second betting round. After the conclusion of the second betting round, players remaining in the game will reveal their private cards. If a player pairs her card with the community card, she will win the pot.

If a player's single private card forms a pair with the community card, she will win the pot. Otherwise, the player with the highest private card wins. We assume that Player 1 is the opponent, while Player 2 and Player 3 are team members. In the adversarial team games, there are some modifications to the payoff structure. If Player 1 wins, she takes all the chips from the pot. If Player 2 or Player 3 wins, the chips contributed by the team members are returned to them, and the chips bet by Player 1 are evenly distributed among each team member.

    \item \textbf{The rule of \emph{Goofspiel}:}
\emph{Goofspiel} is a bidding game.  We adopt a variant version with three cards. Every player has a hand of cards with values $\{1, 2, 3\}$. A third stack of cards, also with values $\{1, 2, 3\}$, is shuffled and placed on the table. At the onset of each round, a neutral referee places a card on the table as the reward for that round. Players bid by selecting a card from their hand, and the player with the highest bid claims the reward. In the event of a tie, the reward is equitably shared among the tying players. After three rounds, all rewards are distributed among the players, contributing to their scores. Assuming that Player 2 and Player 3 form a team, the final score of each team member is calculated by summing and averaging the rewards obtained by Player 2 and Player 3.
\end{itemize}
\end{document}